\documentclass{stacs_proc}

\usepackage{amsmath}
\usepackage{amssymb}

\usepackage{graphics}
\usepackage{epic}
\usepackage{eepic}

\usepackage{stmaryrd}
\usepackage{algorithm}
\usepackage{algorithmic}

\newcommand{\comment}[1]{} 

\newcommand{\set}[2]{\{\,#1\mid#2\,\}} 

\newcommand{\co}{\mbox{\bf co}}
\newcommand{\AC}{{\rm \bf AC}}
\newcommand{\NC}{{\rm \bf NC}}
\renewcommand{\L}{{\rm \bf L}}
\newcommand{\FL}{{\rm \bf FL}}
\newcommand{\NL}{{\rm \bf NL}}

\newcommand{\SL}{{\rm \bf SL}}

\newcommand{\UL}{{\rm \bf UL}}
\newcommand{\FUL}{{\rm \bf FUL}}
\newcommand{\NP}{{\rm \bf NP}}

\newcommand{\Ord}{\mbox{{\sc Ord}}}

\newcommand{\code}{{\it code}}
\newcommand{\Grid}{{\it grid}}
\newcommand{\assign}{\leftarrow}

\begin{document}

%
%
%

\title[The Isomorphism Problem for Planar 3-Connected Graphs is in UL]{The Isomorphism Problem for Planar 3-Connected Graphs is in Unambiguous Logspace}

\author[ref1]{T. Thierauf}{Thomas Thierauf}
\address[ref1]{Fak.\ Elektronik und Informatik, HTW Aalen, 73430 Aalen, Germany}
\email{thomas.thierauf@uni-ulm.de}

\author[ref2]{F. Wagner}{Fabian Wagner}
\address[ref2]{Inst.\ f\"{u}r Theoretische Informatik, Universit\"{a}t Ulm, 89069 Ulm, Germany}
\email{fabian.wagner@uni-ulm.de}

\thanks{Supported by DFG grants Scho 302/7-2 and TO 200/2-1.}

\begin{abstract}
The isomorphism problem for planar graphs is known to be efficiently solvable.
For planar 3-connected graphs, 
the isomorphism problem can be solved by efficient parallel algorithms,
it is in the class~$\AC^1$.

In this paper 
we improve the upper bound for planar 3-connected graphs to unambiguous logspace,
in fact to $\UL \cap \co\UL$.
As a consequence of our method we get that the isomorphism problem for oriented graphs is in $\NL$.
We also show that the problems are hard for $\L$.
\end{abstract}

\maketitle

\stacsheading{2008}{633-644}{Bordeaux}
\firstpageno{633}

\vskip-0.3cm
\section{Introduction}

The graph isomorphism problem (GI) is one of the most challenging problems today.
No polynomial time algorithm is known for it,
even with extended resources like randomization or on quantum computers.
On the other hand,
it is not known to be $\NP$-complete
and there are good reasons to conjecture that it is in fact not complete.

For some restricted classes of graphs, efficient algorithms for GI are known.
For example for trees~\cite{AhoHopUll74} 
or for graphs with bounded degree~\cite{Luks82}.
We are interested in planar graphs and 3-connected graphs.
A graph is 3-connected if it remains connected after deleting two arbitrary vertices.
In 1966, Weinberg~\cite{Weinberg66} presented an $O(n^2)$-algorithm 
for testing isomorphism of planar 3-connected graphs. 
This algorithm was improved and extended by Hopcroft and Tarjan~\cite{HopTar74} 
to an $O(n \log n)$-algorithm
for the planar graph isomorphism problem (planar-GI).
Then Hopcroft and Wong~\cite{HopWon74} showed that it is solvable in linear time.
Since the constant hidden in the linear time bound is very large,
the problem has been reconsidered under a more practical approach~\cite{KukHolCoo04}.
The parallel complexity of planar-GI has been studied by 
Miller and Reif~\cite{MilRei91} and Ramachandran and Reif~\cite{RamRei94}.
They showed that planar-GI 
$\AC^1$-reduces to the 3-connected case and that 3-connected GI is in $\AC^1$.

Grohe and Verbitsky~\cite{GroVer06} gave an alternative way to show  that  planar-GI is in $\AC^1$.
They proved for a class~$\mathcal{G}$ of graphs, that if every graph in~$\mathcal{G}$
is definable in a finite-variable first order logic within 
logarithmic quantifier depth, then the isomorphism problem  for~$\mathcal{G}$  is  in $\AC^1$.
Later Verbitsky~\cite{Verbitsky07} showed that planar 3-connected graphs are definable 
with 15 variables and quantifier depth $O(\log n)$ which leads to a 14-dimensional Weisfeiler-Lehman algorithm. With the reduction of~\cite{MilRei91} one obtains a new $\AC^1$-algorithm for planar-GI.

In the above papers on planar-GI,
the authors consider first 3-connected graphs.
The reason is a result due to Whitney~\cite{Whitney33}
that every planar 3-connected graph has precisely two embeddings on a sphere,
where one embedding is the mirror image of the other.
Moreover,
one can efficiently compute these embeddings.
Weinberg~\cite{Weinberg66} used these embeddings to compute a code for a graph,
such that isomorphic graphs will have the same code.
We call a code with this property a {\em canonical code\/}  for the graph.

Some of the subroutines in the above algorithms have complexity below $\AC^1$.
Allender and Mahajan~\cite{AllMah00}  showed that planarity testing 
is hard for $\L$ and in symmetric logspace, $\SL$.
Since $\SL = \L$~\cite{Reingold05}, planarity testing is complete for logspace.
Furthermore Allender and Mahajan~\cite{AllMah00} showed that a planar embedding can be computed in logspace.
Also the connectivity structure of a (undirected) graph can be computed in logspace~\cite{NisTas95}.
Hence a natural question is whether planar-GI is in logspace.

While this question remains open,
we considerably improve the upper bound for planar-GI for 3-connected graphs in Section~\ref{s:3planar-GI},
namely from $\AC^1$ to unambiguous logspace,  in fact to $\UL \cap \co\UL$.
Like Weinberg, we construct codes for the given graphs.
In order to use only logarithmic space,
our code is constructed via a spanning tree,
which depends on the planar embedding of the graph.
A crucial tool in
the construction of the spanning tree is based on a recent result by 
Bourke, Tewari, and Vinodchandran~\cite{BouTewVin07}
that the reachability problem for planar directed graphs is in $\UL \cap \co\UL$.
They built on work of Reinhard and Allender~\cite{ReiAll00} and Allender, Datta, and Roy~\cite{AllDatRoy05}.
We argue in Section~\ref{s:distances} that their algorithm can be modified to not just solve reachability questions
but to compute distances between nodes in $\UL \cap \co\UL$.

The embedding of a planar graph can be represented as a {\em rotation scheme\/}.
Intuitively this gives the edges in clockwise or counter clockwise order around each node
such that it leads to a planar drawing of the graph.
Rotation schemes have also been considered for non-planar graphs.
We talk of {\em oriented graphs\/} in this case.
We extend our results to the isomorphism problem for oriented graphs.
There one has given two graphs~$G$ and~$H$ and a rotation scheme for each of the graphs.
One has to decide whether there is an isomorphism between~$G$ and~$H$ that respects the rotation schemes.
In Section~\ref{S:OGI} we show that the problem is in $\NL$.

With respect to lower bounds,
GI is known to be hard for DET~\cite{Toran04},
where DET is the class of problems that are $\NC^1$-reducible to the determinant
defined by~\cite{Cook85}.
In fact,
already the isomorphism problem for tournament graphs is hard for DET~\cite{Wagner07}.
We show in Section~\ref{s:hardness}
that the isomorphism problem for planar 3-connected graphs is hard for logspace.

\vskip-0.3cm
\section{Preliminaries}

Basically, $\L$ and $\NL$ are the classes of languages computable by a deterministic 
and nondeterministic logspace bounded Turing machine, respectively. 
A nondeterministic Turing machine is called {\em unambiguous\/},
if it has at most one accepting computation on any input.
The class of languages computable by
unambiguous logspace bounded Turing machines is denoted by $\UL$.
$\NL$ is known to be closed under complement~\cite{Immerman88,Szelepcsenyi88}, 
but it is open for $\UL$.

The functional version of~$\L$ is denoted by~$\FL$.
It is known that $\FL$-functions are closed under composition, i.e.\ $\FL \circ \FL = \FL$.
The proof goes by recomputing bits of the function value of the first function
each time such a bit is needed by the second function.
The same argument works when we consider functions that are computed by unambiguous
logspace bounded Turing machines.
If we  call the class $\FUL$, then this says that $\FUL \circ \FUL = \FUL$.
We need a further property of $\UL$:
\begin{lemma}\label{le:log-UL}
$\L^{\UL \cap \co\UL} = \UL \cap \co\UL$.
\end{lemma}
\begin{proof}
Let $M$ be a logspace oracle Turing machine with oracle~$A \in \UL \cap \co\UL$.
Let $M_0,M_1$ be (nondeterministic) unambiguous logspace Turing machines
such that $L(M_0) = \overline{A}$ and $L(M_1) = A$.
An unambiguous logspace Turing machine~$M'$ for $L(M,A)$ works as follows
on input~$x$:
\begin{quote}
Simulate~$M$ on input~$x$.
If~$M$ asks an oracle question~$y$,
then nondeterministically guess whether the answer is~$0$ or~$1$.
\begin{itemize}
\item
If the guess is answer~$0$,
then simulate $M_0$ on input~$y$.
If $M_0$ accepts, then continue the simulation of~$M$ with oracle answer~$0$.
If $M_0$ rejects then reject and halt.
\item
If the guess is answer~$1$,
then simulate $M_1$ on input~$y$.
If $M_1$ accepts, then continue the simulation of~$M$ with oracle answer~$1$.
If $M_1$ rejects then reject and halt.
\end{itemize}
Finally accept iff $M$ accepts.
\end{quote}
Note that~$M'$ is unambiguous because~$M_0$ and~$M_1$ are unambiguous
and of the two guessed oracle answers always exactly one guess is correct.
\end{proof}

Let~$G=(V,E)$ be an undirected graph with vertices~$V=V(G)$ 
and edges~$E=E(G)$. Let~$G - \{v\}$ denote the induced subgraph of~$G$ on~$V(G) \setminus \{v\}$.
The \emph{neighbours} of~$v \in V$ are~$\Gamma(v) = \set{u}{(v,u) \in E}$.
By $E_v$ we denote the edges going from~$v$ to its neighbors,
$E_v = \set{(v,u)}{u \in \Gamma(v)}$.
By $d(u,v)$ we denote the distance between nodes~$u$ and~$v$ in~$G$,
which is the length of a shortest path from~$u$ to~$v$ in~$G$.

A graph is \emph{connected} if there is a path between any two vertices in $G$.
A vertex $v \in V$ is an \emph{articulation point} if $G - \{v\}$ is  not connected.
A pair of vertices $u,v \in V$ is a \emph{separation pair} if $G - \{u,v\}$ is not connected.
A \emph{biconnected graph} contains no articulation points.
A 3-\emph{connected graph} contains no separation pairs.

A \emph{rotation scheme} for a graph~$G$ is a set~$\rho$ of permutations,
$\rho = \set{\rho_v}{v\in V}$,
where~$\rho_v$ is a permutation on~$E_v$ that has only one cycle
(which is called a rotation).
Let $\rho^{-1}$  be the set of inverse rotations,
$\rho^{-1} = \set{\rho_v^{-1}}{v\in V}$.
A rotation scheme~$\rho$ describes an embedding of graph~$G$ in the plane.
We call~$G$ together with~$\rho$ an {\em oriented graph\/}.
If the embedding is planar, we call~$\rho$ a {\em planar rotation scheme\/}.
Note that in this case~$\rho^{-1}$ is a planar rotation scheme as well.
Allender and Mahajan~\cite{AllMah00} showed that a planar rotation scheme for a planar graph
can be computed in logspace.

If a planar graph is in addition 3-connected,
then there exist precisely two planar rotation schemes~\cite{Whitney33},
namely some planar rotation scheme~$\rho$ and its inverse $\rho^{-1}$.
This is a crucial property in our isomorphism test.

\vskip-0.3cm
\section{Planar 3-Connected Graph Isomorphism}\label{s:3planar-GI}

In this section we prove the following theorem.

\begin{theorem}\label{th:iso-planar3}
The isomorphism problem for planar, $3$-connected graphs is in $\UL \cap \co\UL$.
\end{theorem}

In 1966, Weinberg~\cite{Weinberg66} presented an $O(n^2)$ algorithm for 
testing isomorphism of planar 3-connected graphs.
The algorithm computes a {\em canonical form\/} for each of the two graphs.
This is a coding of graphs
such that these codings are equal iff the two graphs are isomorphic.
For a 3-connected graph~$G$, 
the algorithm starts by constructing a code for every edge of~$G$ and 
any of the two rotation schemes.
Of all these codes, the lexicographical smallest one is the code for~$G$.

For a designated edge $(s,t)$ and a rotation scheme~$\rho$ for~$G$,
the code is constructed roughly as follows.
Every undirected edge is considered as two directed edges.
Now one can define an Euler tour based on some rules for selecting the next edge.
Basically, the rules distinguish between the case whether a vertex or  edge was already visited or not.
The next edge to consider is chosen to the left or right of the active edge according to~$\rho$.
Define edge $(s,t)$ to be the start of the tour.
The code consists of the nodes as they appear on the tour,
where the names are replaced by the order of their first appearance on the tour.
That is,
the code starts with $(1, 2)$ for the edge $(s,t)$
and every later occurrence of~$s$ or~$t$ on the tour is replaced by~1 or~2, respectively.

Weinberg's algorithm doesn't work in logspace, 
because one has to store the vertices and edges already visited.
We show how to construct a different code in $\UL$.
Let $(s,t)$ be a designated edge and~$\rho$ be a rotation scheme for~$G$.
Our construction makes three steps.

\begin{enumerate}
\item 
First 
we compute a canonical spanning tree~$T$ for~$G$.
This is a spanning tree which depends on $(s,t)$, $\rho$, and~$G$,
but not on the way these inputs are represented.
\item
Next 
we construct a canonical list~$L$ of all edges of~$G$.
To do so, 
we traverse~$T$ and enumerate the edges of~$T$ and their neighbor edges according to~$\rho$.
The list~$L$ does not depend on the representation of~$G$, $\rho$ or~$T$.
\item
Finally
we rename the vertices depending on the position of their first occurrence in the list~$L$
and get a code word for~$G$ with respect to $(s,t)$ and~$\rho$.
\end{enumerate}

We will see that the spanning tree in step~1 can be computed in
(the functional version of) $\UL \cap \co\UL$.
The list and the renaming in step~2 and step~3 can be computed in logspace, $\L$.
Therefore
the composition of the three steps is in $\UL \cap \co\UL$.

The overall algorithm has to decide
whether two given graphs~$G$ and~$H$ are isomorphic.
To do so
we fix $(s,t)$ and~$\rho$ for $G$
and cycle through all edges of~$H$ as designated edge and the two possible
permutation schemes of~$H$.
Then~$G$ and~$H$ are isomorphic iff we find a code for~$H$ that matches the code for~$G$.
It is not hard to see that this outer loop is in logspace.
Therefore the isomorphism test stays in $\UL \cap \co\UL$.

\subsection*{Step 1: Construction of a canonical spanning tree}

We show that the following problem can be solved in unambiguous logspace.
\begin{itemize}
\item
\emph{Input:}
An undirected graph $G = (V,E)$, a rotation scheme~$\rho$ for~$G$,
and a designated edge $(s,t) \in E$.
\item
\emph{Output:} A canonical spanning tree~$T \subseteq E$ of~$G$.
\end{itemize}

Recall that by a canonical spanning tree we mean that~$T$ 
does not depend on the input representation of~$\rho$ or~$G$,
any representation will result in the same spanning tree~$T$.

The idea to construct the spanning tree is to traverse~$G$ with a breath-first search
starting at node~$s$.
The neighbors of a node are visited in the order given by the rotation scheme~$\rho$.
Since the algorithm should work in logspace, 
we cannot afford to store all the nodes that we already visited, as in a standard breath-first search.
We get around this problem by working with distances between nodes.

We start with the nodes at distance~1 from~$s$.
That is, write $(s,v)$ on the output tape, for all $v \in \Gamma(s)$.
Now let $d \geq 2$ and assume that we have already constructed~$T$ up to nodes
at distance $\leq d-1$ to~$s$.
Then we  consider the nodes at distance~$d$ from~$s$.
Let~$w$ be a node with $d(s,w) = d$.
We have to connect~$w$ to the tree constructed so far.
We do so by computing a shortest path from~$s$ to $w$.
Ambiguities are resolved by using the first feasible edge according to~$\rho$.
We start with $(s,t)$ as the active edge $(u,v)$.
\begin{itemize}
\item
If $d(u,w) > d(v,w)$, 
then $(u,v)$ is the first  edge encountered that is on a shortest path from~$u$ to~$w$.
Therefore we go from~$u$ to~$v$ and start searching the next edge from~$v$.
As starting edge we take the successor of $(v,u)$.
That is, $\rho_v(v,u)$ is the new active edge.
\item
If $d(u,w) \leq d(v,w)$, 
then $(u,v)$ is not on a shortest path from~$u$ to~$w$.
Then we proceed with $\rho_u(u,v)$ as the new active edge.
\end{itemize}
After $d-1$ steps in direction of~$w$
the node~$v$ of the active edge $(u,v)$ is a predecessor of~$w$ on a shortest path from~$s$ to~$w$.
Then we write $(v,w)$ on the output tape.
The following pseudo-code summarizes the algorithm.

\begin{tabbing}
xxx\=xxx\=xxx\=xxx\=xxx\=xxx\= \kill
\> {\bf for all} $v \in \Gamma(s)$ {\bf do} ~~{\bf output} $(s,v)$\\
\> {\bf for } $d \assign 2$ {\bf to} $n-1$ {\bf do}\\
\> \>  {\bf for all} $w \in V$ such that $d(s,w) = d$ {\bf do}\\
\>\> \> $(u,v) \assign (s,t)$\\
\>\> \> {\bf for } $k \assign 1$ {\bf to} $d-1$ {\bf do}\\
\> \> \> \> {\bf while}    $d(u,w) \leq d(v,w)$ {\bf do} $(u,v) \assign \rho_u(u,v)$ \\
\> \> \> \>  $(u,v) \assign \rho_v(v,u)$\\
\> \> \> {\bf output} $(v,w)$
\end{tabbing}

The spanning tree~$T$ is canonical because its construction depends only  on~$\rho$,
edge $(s,t)$, and edge set~$E$.
The following figure shows an example of a spanning tree~$T$ for a graph~$G$ 
with rotation function $\rho$
which arranges the edges in clockwise order around each vertex.

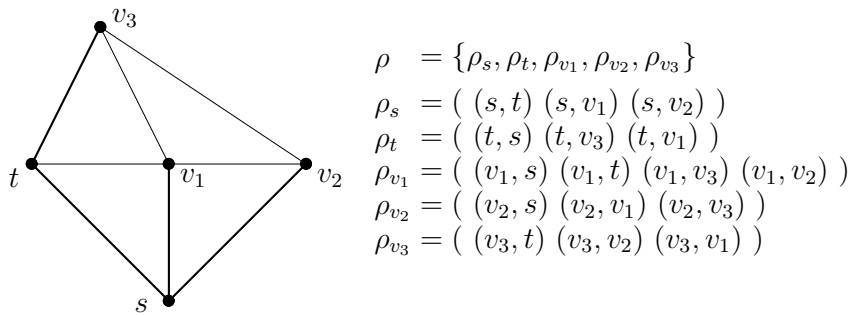
\vspace{-2ex}

\begin{center}
\begin{figure}[htbp]
\scalebox{0.95}{ 
\setlength{\unitlength}{0.00083333in}
\begingroup\makeatletter\ifx\SetFigFont\undefined
\def\x#1#2#3#4#5#6#7\relax{\def\x{#1#2#3#4#5#6}}%
\expandafter\x\fmtname xxxxxx\relax \def\y{splain}%
\ifx\x\y   
\gdef\SetFigFont#1#2#3{%
  \ifnum #1<17\tiny\else \ifnum #1<20\small\else
  \ifnum #1<24\normalsize\else \ifnum #1<29\large\else
  \ifnum #1<34\Large\else \ifnum #1<41\LARGE\else
     \huge\fi\fi\fi\fi\fi\fi
  \csname #3\endcsname}%
\else
\gdef\SetFigFont#1#2#3{\begingroup
  \count@#1\relax \ifnum 25<\count@\count@25\fi
  \def\x{\endgroup\@setsize\SetFigFont{#2pt}}%
  \expandafter\x
    \csname \romannumeral\the\count@ pt\expandafter\endcsname
    \csname @\romannumeral\the\count@ pt\endcsname
  \csname #3\endcsname}%
\fi
\fi\endgroup
{\renewcommand{\dashlinestretch}{30}
\begin{picture}(6346,2095)(0,-10)
\put(2400,424){\makebox(0,0)[lb]{{\SetFigFont{12}{14.4}{rm}$\rho_{v_3}$}}}
\put(2400,649){\makebox(0,0)[lb]{{\SetFigFont{12}{14.4}{rm}$\rho_{v_2}$}}}
\put(2400,874){\makebox(0,0)[lb]{{\SetFigFont{12}{14.4}{rm}$\rho_{v_1}$}}}
\put(2400,1099){\makebox(0,0)[lb]{{\SetFigFont{12}{14.4}{rm}$\rho_t$}}}
\put(2700,1324){\makebox(0,0)[lb]{{\SetFigFont{12}{14.4}{rm}$= (~(s,t)~(s,v_1)~(s,v_2)~)$}}}
\put(2700,1099){\makebox(0,0)[lb]{{\SetFigFont{12}{14.4}{rm}$= (~(t,s)~(t,v_3)~(t,v_1)~)$}}}
\put(2700,874){\makebox(0,0)[lb]{{\SetFigFont{12}{14.4}{rm}$= (~(v_1,s)~(v_1,t)~(v_1,v_3)~(v_1,v_2)~)$}}}
\put(2700,649){\makebox(0,0)[lb]{{\SetFigFont{12}{14.4}{rm}$= (~(v_2,s)~(v_2,v_1)~(v_2,v_3)~)$}}}
\put(2700,424){\makebox(0,0)[lb]{{\SetFigFont{12}{14.4}{rm}$= (~(v_3,t)~(v_3,v_2)~(v_3,v_1)~)$}}}
\put(2400,1324){\makebox(0,0)[lb]{{\SetFigFont{12}{14.4}{rm}$\rho_s$}}}
\put(2400,1624){\makebox(0,0)[lb]{{\SetFigFont{12}{14.4}{rm}$\rho$}}}
\put(2700,1624){\makebox(0,0)[lb]{{\SetFigFont{12}{14.4}{rm}$= \{ \rho_s, \rho_t, \rho_{v_1}, \rho_{v_2}, \rho_{v_3} \}$}}}
\put(1050,124){\blacken\ellipse{72}{72}}
\put(1050,124){\ellipse{72}{72}}
\put(1950,1024){\blacken\ellipse{72}{72}}
\put(1950,1024){\ellipse{72}{72}}
\put(1050,1024){\blacken\ellipse{72}{72}}
\put(1050,1024){\ellipse{72}{72}}
\put(150,1024){\blacken\ellipse{72}{72}}
\put(150,1024){\ellipse{72}{72}}
\put(600,1924){\blacken\ellipse{72}{72}}
\put(600,1924){\ellipse{72}{72}}
\thicklines
\path(1050,124)(1050,1024)
\path(1050,124)(1950,1024)
\path(1050,124)(150,1024)
\thinlines
\path(600,1924)(1050,1024)
\path(600,1924)(1950,1024)
\thicklines
\path(600,1924)(150,1024)
\thinlines
\path(150,1024)(1050,1024)
\path(1050,1024)(1950,1024)
\put(1125,874){\makebox(0,0)[lb]{{\SetFigFont{12}{14.4}{rm}$v_1$}}}
\put(675,1924){\makebox(0,0)[lb]{{\SetFigFont{12}{14.4}{rm}$v_3$}}}
\put(0,874){\makebox(0,0)[lb]{{\SetFigFont{12}{14.4}{rm}$t$}}}
\put(825,49){\makebox(0,0)[lb]{{\SetFigFont{12}{14.4}{rm}$s$}}}
\put(2025,874){\makebox(0,0)[lb]{{\SetFigFont{12}{14.4}{rm}$v_2$}}}
\end{picture}
}
}
\label{F:Pl3ConRotFkt}
\end{figure}
\end{center}


Except for the computation of the distances,
the algorithm works in logspace.
We have to store the values of~$d$, $k$, $u$ and~$v$, and the position of~$w$,
plus some extra space for doing calculations.
We show in Theorem~\ref{th:dist} below that the distances can be computed in $\UL \cap \co\UL$.
By Lemma~\ref{le:log-UL}
the canonical spanning tree can be computed in $\UL \cap \co\UL$.

\subsection*{Step 2: Computation of a canonical list of all edges}

We show that the following problem can be solved in logspace.
\begin{itemize}
\item
\emph{Input:}
An undirected graph $G = (V,E)$, a rotation scheme~$\rho$ for~$G$,
a spanning tree~$T \subseteq E$ of~$G$, and a designated edge $(s,t) \in T$.
\item
\emph{Output:} A canonical list~$L$  of all edges in~$E$.
\end{itemize}

Recall that by a canonical list we mean that
the order of the edges as they appear in $L$
does not depend on the input representation of~$\rho$, $G$ or~$T$,
any representation will result in the same list.

The idea is to traverse the spanning tree in a depth-first manner. 
At each vertex visit all incident edges in breath-first manner according to $\rho$
until the next edge contained in the spanning tree is reached.

We start the traversal with edge $(s,t)$ as the active edge $(u,v)$.
We write $(u,v)$ on the output tape and then compute the next active edge as follows:
\begin{itemize}
\item
If $(u,v) \in T$ then we walk depth-first in~$T$ from~$u$ to~$v$,
consider the edge $(v,u)$ and take its successor according to~$\rho_v$,
i.e., $\rho_v(v,u)$ is the new active edge.
\item
If $(u,v) \not\in T$ then we proceed breath-first with $\rho_u(u,v)$
as the new active edge.
\end{itemize}
This step is repeated until the active edge is again $(s,t)$.
Then we have traversed all edges in~$E$.
Every undirected edge is encountered exactly twice, once in each direction.
The following pseudo-code summarizes the algorithm.

\begin{tabbing}
xxx\=xxx\=xxx\=xxx\=xxx\=xxx\= \kill
\> $(u,v) \assign (s,t)$\\
\> {\bf repeat}\\
\> \>  {\bf output} $(u,v)$\\
\> \>  {\bf if} $(u,v) \in T$ {\bf then} $(u,v) \assign \rho_v(v,u)$\\
\> \>  {\bf else} $(u,v) \assign \rho_u(u,v)$\\
\> {\bf until} $(u,v) =  (s,t)$
\end{tabbing}

Clearly, the algorithm works in logspace.
The list~$L$ is canonical because its construction depends only  on~$\rho$,
edge $(s,t)$, and sets~$E$ and~$T$.
Since~$T$ is canonical as well,
$L$ depends actually only on~$\rho$, $(s,t)$, and~$E$.
The following figure 
shows an example for $L$.

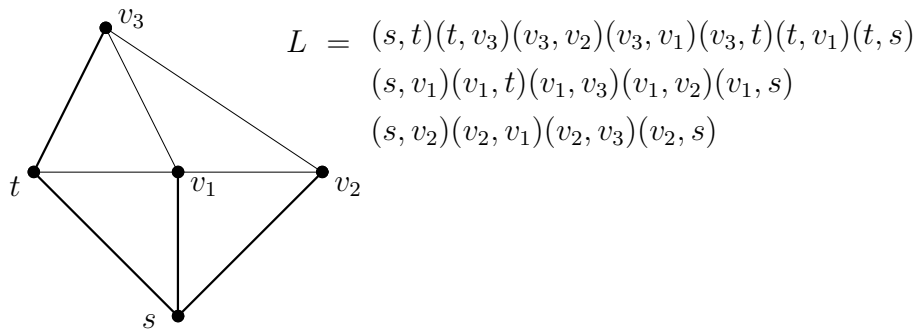
\begin{center}
\scalebox{1}{ 
\setlength{\unitlength}{0.00083333in}
\begingroup\makeatletter\ifx\SetFigFont\undefined
\def\x#1#2#3#4#5#6#7\relax{\def\x{#1#2#3#4#5#6}}%
\expandafter\x\fmtname xxxxxx\relax \def\y{splain}%
\ifx\x\y   
\gdef\SetFigFont#1#2#3{%
  \ifnum #1<17\tiny\else \ifnum #1<20\small\else
  \ifnum #1<24\normalsize\else \ifnum #1<29\large\else
  \ifnum #1<34\Large\else \ifnum #1<41\LARGE\else
     \huge\fi\fi\fi\fi\fi\fi
  \csname #3\endcsname}%
\else
\gdef\SetFigFont#1#2#3{\begingroup
  \count@#1\relax \ifnum 25<\count@\count@25\fi
  \def\x{\endgroup\@setsize\SetFigFont{#2pt}}%
  \expandafter\x
    \csname \romannumeral\the\count@ pt\expandafter\endcsname
    \csname @\romannumeral\the\count@ pt\endcsname
  \csname #3\endcsname}%
\fi
\fi\endgroup
{\renewcommand{\dashlinestretch}{30}
\begin{picture}(6265,2095)(0,-10)
\put(1050,124){\blacken\ellipse{72}{72}}
\put(1050,124){\ellipse{72}{72}}
\put(1950,1024){\blacken\ellipse{72}{72}}
\put(1950,1024){\ellipse{72}{72}}
\put(1050,1024){\blacken\ellipse{72}{72}}
\put(1050,1024){\ellipse{72}{72}}
\put(150,1024){\blacken\ellipse{72}{72}}
\put(150,1024){\ellipse{72}{72}}
\put(600,1924){\blacken\ellipse{72}{72}}
\put(600,1924){\ellipse{72}{72}}
\thicklines
\path(1050,124)(1050,1024)
\path(1050,124)(1950,1024)
\path(1050,124)(150,1024)
\thinlines
\path(600,1924)(1050,1024)
\path(600,1924)(1950,1024)
\thicklines
\path(600,1924)(150,1024)
\thinlines
\path(150,1024)(1050,1024)
\path(1050,1024)(1950,1024)
\put(1125,874){\makebox(0,0)[lb]{{\SetFigFont{12}{14.4}{rm}$v_1$}}}
\put(675,1924){\makebox(0,0)[lb]{{\SetFigFont{12}{14.4}{rm}$v_3$}}}
\put(0,874){\makebox(0,0)[lb]{{\SetFigFont{12}{14.4}{rm}$t$}}}
\put(825,49){\makebox(0,0)[lb]{{\SetFigFont{12}{14.4}{rm}$s$}}}
\put(2250,1174){\makebox(0,0)[lb]{{\SetFigFont{12}{14.4}{rm}$(s,v_2)(v_2,v_1)(v_2,v_3)(v_2,s)$}}}
\put(2250,1474){\makebox(0,0)[lb]{{\SetFigFont{12}{14.4}{rm}$(s,v_1)(v_1,t)(v_1,v_3)(v_1,v_2)(v_1,s)$}}}
\put(2250,1774){\makebox(0,0)[lb]{{\SetFigFont{12}{14.4}{rm}$(s,t)(t,v_3)(v_3,v_2)(v_3,v_1)(v_3,t)(t,v_1)(t,s)$}}}
\put(1725,1774){\makebox(0,0)[lb]{{\SetFigFont{12}{14.4}{rm}$L~~=$}}}
\put(2025,874){\makebox(0,0)[lb]{{\SetFigFont{12}{14.4}{rm}$v_2$}}}
\end{picture}
}}
\end{center}

\subsection*{Step 3: Renaming the vertices}

The last step is to rename the vertices in  the list~$L$ such that they become independent of the
names they have in~$G$.
This is achieved as follows:
consider the first occurrence (from left) of node~$v$ in~$L$.
Let $k-1$ be the number of pairwise different nodes to the left of~$v$.
Then all occurrences of~$v$ are replaced by $k$.
Recall that $L$ starts with the edge $(s,t)$.
Hence all occurrences of~$s$ get replaced by $1$,
all occurrences of~$t$ get replaced by $2$,
and so on.
Call the new list $\code(G,\rho,s,t)$.

Given~$L$ as input,
the list $\code(G,\rho,s,t)$ can be computed in logspace.
We start with the first node~$v$ (which is~$s$) and  a counter~$k$, 
that counts the number of different nodes we have seen so far.
In the beginning, we set $k=1$.
\begin{itemize}
\item
If~$v$ occurs for the first time, than we output~$k$ and increase~$k$ by~1.
\item
If~$v$ occurs already to the left of the current position,
then we have to determine the number, $v$ got at its first occurrence.
To do so,
we determine the first occurance of~$v$ and 
then count the number of different nodes to the left of of~$v$ at its first occurance.
It is not hard to see that this can be done in logspace.
\end{itemize}
Then we go to the next node in~$L$. 
Consider the example from above. 
The code constructed from list~$L$ for~$G$ is as follows.

\[
\begin{array}{rccccccccc}
 L =                    & (s,t) & (t,v_3)  &  (v_3,v_2)  &  (v_3,v_1)  &  (v_3,t)  &  (t,v_1)  &  (t,s)  \\
 \code(G,\rho,s,t) =    & (1,2) &  (2,3)    &  (3,4)      &  (3,5)      &  (3,2)    &  (2,5)    &  (2,1)  \\[1ex]
\text{sequel of }  L   &         (s,v_1) &  (v_1,t)  &  (v_1,v_3)  &  (v_1,v_2)  &  (v_1,s)  \\
\text{sequel of } \code &         (1,5)    &  (5,2)    &  (5,3)      &  (5,4)      &  (5,1)  \\[1ex]
\text{sequel of } L   &          (s,v_2)  &  (v_2,v_1)  &  (v_2,v_3)  &  (v_2,s)  \\
\text{sequel of } \code &         (1,4)    &  (4,5)      &  (4,3)      &  (4,1) 
\end{array}
\]

The renaming algorithm works in logspace.
It remains to argue that the new names of the nodes are independent of their names in~$G$.
Let~$H$ be a graph which is isomorphic to~$G$,
and let $\varphi$ be an isomorphism between~$G$ and~$H$.
Note that $\rho \circ \varphi$ is a rotation scheme for~$H$.
Consider the computation of the code for graph~$H$
with rotation scheme $\rho \circ \varphi$  and designated edge $(\varphi(s),\varphi(t))$.
The spanning tree computed in step~1 will be $\varphi(T)$
and the list computed in step~2 will be $\varphi(L)$.
Now the above renaming procedure will give the same number
to node~$v$ in~$L$ and to node $\varphi(v)$ in $\varphi(L)$.
For example nodes~$\varphi(s)$  and ~$\varphi(t)$ will get number~1 and~2, respectively.
It follows that $\code(G,\rho,s,t) = \code(H, \rho \circ \varphi, \varphi(s),\varphi(t))$.
We summarize:

\begin{theorem}
Let $G$ and $H$ be connected, undirected graphs,
let $\rho_G$ be a rotation scheme for~$G$ and $(s,t)$ be an edge in~$G$.
Then 
$G$ and $H$ are isomorphic iff
there exists a rotation scheme~$\rho_H$ for~$H$ and an edge $(u,v)$ in~$H$ such that
$\code(G,\rho_G,s,t) = \code(H, \rho_H, u,v)$.
\end{theorem}

This completes the proof of Theorem~\ref{th:iso-planar3}
except for the complexity bound on computing distances in planar graphs.
This is done in the next section.

\vskip-0.3cm
\section{Computing Distances in Planar Graphs}\label{s:distances}

We show that distances in planar graphs can be computed in unambiguous logspace.

\begin{theorem}\label{th:dist}
The distance between any two vertices in a planar graph 
can be computed in $\UL \cap \co\UL$.
\end{theorem}

Bourke, Tewari, and Vinodchandran~\cite{BouTewVin07} showed that
the reachability problem for planar directed graphs is in $\UL \cap \co\UL$.
Their algorithm is essentially based on  two results:

\begin{enumerate}
\item
Allender, Datta, and Roy~\cite{AllDatRoy05} showed that the
reachability problem for planar directed graphs can be reduced  to grid graph reachability.
Grid graphs are graphs who's vertices can be identified with the grid points in a
2-dimensional grid with the edges connecting only the direct
horizontal or vertical neighbors.
\item
Reinhard and Allender~\cite{ReiAll00} showed that the $\NL$-complete reachability problem for directed graphs
is in $\UL \cap \co\UL$ if there is a logspace computable weight function for the edges
such that for every pair of vertices~$u$ and~$v$,
if there is a path from~$u$ to~$v$, 
then there is a unique minimum weight shortest path between~$u$ and~$v$.
\end{enumerate}
Bourke, Tewari, and Vinodchandran~\cite{BouTewVin07}
provide such a weight function for grid graphs.
Therefore the reachability problem for planar directed graphs is in $\UL \cap \co\UL$.

We modify this algorithm in order to determine distances between nodes
in the given planar graph~$G$.
This is adapted from the Reinhard-Allender algorithm applied to the weighted grid graph
computed from~$G$.
Here, we only describe the changes that have to be made in the cited references.

We start by considering the reduction 
from reachability for a planar graph~$G$ to a grid graph~$G_{\Grid}$~\cite{AllDatRoy05}.
The reduction from~$G$ to~$G_{\Grid}$ is a special combinatorial embedding
that introduces only degree 2 nodes, 
thereby it preserves the exact number of paths between any two original vertices. 
Vertices in~$G$ are replaced by directed cycles and edges in~$G$ are replaced by paths 
such that they can be embedded into a grid.
For our purpose it suffices to note that one can modify the construction 
and {\em mark the original edges of\/}~$G$ in~$G_{\Grid}$.
Hence if we consider paths in $G_{\Grid}$ and count only the marked edges, 
we get distances in~$G$.

The next step is to define a weight function such that shortest paths in $G_{\Grid}$ 
with respect to marked edges are unique.
Bourke, Tewari, and Vinodchandran~\cite{BouTewVin07} defined the following weight function.
For an edge~$e$ let
\[
w_0(e) = 
\begin{cases}
n^4, & \text{if $e$ is an east or west edge},\\
n^4 + i, & \text{if $e$ is a north edge in column $i$},\\
n^4 - i, & \text{if $e$ is a south edge in column $i$}.
\end{cases}
\]
Let~$p$ be a  path in $G_{\Grid}$.
The weight~$w_0(p)$ is the sum of the weights of the edges on~$p$
and can be written as $a + bn^4$.
Clearly, $b$ is the number of edges on~$p$.
Also, it is easy to see that if another  path~$p'$ with weight $w_0(p') = a' + b' n^4$ has the same weight as~$p$,
i.e..\ $w_0(p) = w_0(p')$, then $a = a'$ and $b = b'$.
This enforces that when we consider shortest paths between two nodes,
these paths must have the same number of edges.
The crucial part now is the value of~$a$.
Let $p$ and~$p'$ be different simple paths connecting the same two vertices.
Then Bourke, Tewari, and Vinodchandran~\cite{BouTewVin07} showed that $a \not= a'$.
It follows that the minimum weight path with respect to~$w_0$  is always unique.

Now we modify the weight function in order to give priority to  the marked edges.
That is, we define
\[
w(e) = 
\begin{cases}
w_0(e) + n^8, & \text{if $e$ is marked},\\
w_0(e), & \text{otherwise}.
\end{cases}
\]
Clearly,
minimum weight paths must minimize the number of marked edges.
The next parameter to minimize is the number of all edges on a path.
Finally,
by the same argument as above,
the $a$-values of different simple paths that connect the same two vertices will be different.
It follows  that the minimum weight path with respect to~$w$  is always unique.

Reinhard and Allender~\cite{ReiAll00} extended the counting technique
of Immerman~\cite{Immerman88} and Szelepcs{\'e}nyi~\cite{Szelepcsenyi88}.
In addition to the number of nodes within distance~$k$ from some start node~$s$,
they also sum up the length of the shortest paths to these nodes.
If the shortest paths are unique
then they show that the predicate $d(s,v) \leq k$ is in $\UL \cap \co\UL$.
The distance~$d$ now refers to~$G_{\Grid}$ because this is the input of the algorithm.
By augmenting the algorithm with a counter for marked edges
we also can refer to distances in~$G$ by construction of the weight function~$w$.
This suffices for our purpose
because by several invocations of this procedure with different~$k$'s
we can determine $d(s,v)$ for any~$s$ and~$v$ in $\UL \cap \co\UL$, 
where~$d$ is the distance in~$G$.

\section{Oriented Graph Isomorphism}\label{S:OGI}

In the previous sections we have considered planar graphs,
where the planar embedding is provided by a rotation scheme.
It is also interesting to consider {\em arbitrary\/} (undirected) graphs 
with a rotation scheme that induces some orientation,  i.e.\ cyclic order, on the edges.
In the {\em isomorphism problem for oriented graphs\/} we have given two graphs, each with a rotation scheme.
One has to decide whether there is an isomorphism between the graphs that respects the orientation.

Miller and Reif~\cite{MilRei91} proved
that the isomorphism problem for oriented graphs is in $\AC^1$.
We improve the complexity bound to $\NL$.
The proof goes along the same lines as for planar-GI:
compute a canonical form for each of the graphs according to the given rotation schemes
such that precisely in the isomorphic case, these canonical forms are equal.

\begin{theorem}\label{T:OGI}
The oriented graph isomorphism problem is in $\NL$.
\end{theorem}

It suffices to analyse the complexity of computing a canonical form for  a graph~$G$ and a rotation scheme~$\rho$.
If $G$ is not connected, 
then we determine  the connected components in logspace~\cite{NisTas95,Reingold05} and compute 
canonical forms for each of them. 
Then we sort these canonical forms lexicographically and write them onto the output tape. 
Thus, we may assume that~$G$ is connected.

The three steps to compute a canonical form for a planar graph were all in logspace,
except for the subroutine to compute distances,
which was in $\UL \cap \co\UL$.
Without planarity,
the best upper bound for computing distances in a graph is $\NL$:
to determine if $d(u,v) \leq k$ simply guess a path of length $\leq k$ from~$u$ to~$v$.
This proves Theorem~\ref{T:OGI}.

\section{Hardness of Planar 3-Connected GI}\label{s:hardness}

Lindell~\cite{Lindell92} proved that tree isomorphism (TI) is in $\L$.
In fact, TI is complete for $\L$~\cite{JennerEtAl03}.
Since trees are planar graphs,
planar-GI is hard for $\L$.
We show that the problem remains hard for $\L$
even when restricted to planar 3-connected graphs.
All the hardness and completeness results in this section 
are with respect to  $\AC^0$-many-one reductions.

\begin{theorem}\label{T:Planar3ConnGIHardforL}
Planar 3-connected graph isomorphism is hard for $\L$. 
\end{theorem}

We reduce from the known $\L$-complete problem $\Ord$
which is defined as follows.

\begin{quote}
{\bf Order between Vertices} ($\Ord$)\\
Input:
a directed graph $G = (V,E)$ that is a line, and $s,t \in V$.\\
Decide
whether $s < t$ in the total order induced on~$V$ by~$G$.
\end{quote}

We first describe the reduction from $\Ord$ to TI~\cite{JennerEtAl03}.
Let $v_1, \dots, v_n$ be the nodes of~$G$ in the order they appear on the line in~$G$.
In particular, $v_1$ is the unique node with in-degree~0
and $v_n$ is the unique node with out-degree~0.
Let $s = v_i$ and $t = v_j$.
W.l.o.g.\ assume that $i \not= n$ 
(otherwise map the instance to a non-isomorphic pair of trees).
The (undirected) tree~$T$ constructed from~$G$
has two copies $u_1, \dots, u_n$ and $w_1, \dots, w_n$ of the line of~$G$,
and there is an additional node~$r$ that is connected to~$u_1$ and~$w_1$.
Up to this point, we have constructed one long line.
Now the trick is to interrupt this line:
take out the edge $(u_i,u_{i+1})$ and instead put the edge $(w_i,u_{i+1})$.
Let~$T$ be the resulting tree.

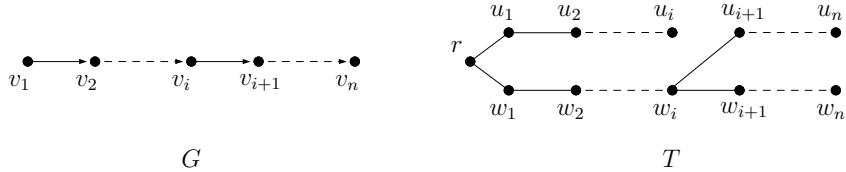
\begin{center}
\scalebox{0.8}{ 
\setlength{\unitlength}{0.00083333in}
\begingroup\makeatletter\ifx\SetFigFont\undefined
\def\x#1#2#3#4#5#6#7\relax{\def\x{#1#2#3#4#5#6}}%
\expandafter\x\fmtname xxxxxx\relax \def\y{splain}%
\ifx\x\y   
\gdef\SetFigFont#1#2#3{%
  \ifnum #1<17\tiny\else \ifnum #1<20\small\else
  \ifnum #1<24\normalsize\else \ifnum #1<29\large\else
  \ifnum #1<34\Large\else \ifnum #1<41\LARGE\else
     \huge\fi\fi\fi\fi\fi\fi
  \csname #3\endcsname}%
\else
\gdef\SetFigFont#1#2#3{\begingroup
  \count@#1\relax \ifnum 25<\count@\count@25\fi
  \def\x{\endgroup\@setsize\SetFigFont{#2pt}}%
  \expandafter\x
    \csname \romannumeral\the\count@ pt\expandafter\endcsname
    \csname @\romannumeral\the\count@ pt\endcsname
  \csname #3\endcsname}%
\fi
\fi\endgroup
{\renewcommand{\dashlinestretch}{30}
\begin{picture}(6807,1367)(0,-10)
\put(3756,1196){\makebox(0,0)[lb]{{\SetFigFont{12}{14.4}{rm}$u_1$}}}
\put(4281,1196){\makebox(0,0)[lb]{{\SetFigFont{12}{14.4}{rm}$u_2$}}}
\put(5031,1196){\makebox(0,0)[lb]{{\SetFigFont{12}{14.4}{rm}$u_i$}}}
\put(5556,1196){\makebox(0,0)[lb]{{\SetFigFont{12}{14.4}{rm}$u_{i+1}$}}}
\put(6306,1196){\makebox(0,0)[lb]{{\SetFigFont{12}{14.4}{rm}$u_n$}}}
\put(3900,1099){\blacken\ellipse{72}{72}}
\put(3900,1099){\ellipse{72}{72}}
\put(3600,874){\blacken\ellipse{72}{72}}
\put(3600,874){\ellipse{72}{72}}
\put(3900,649){\blacken\ellipse{72}{72}}
\put(3900,649){\ellipse{72}{72}}
\put(4425,1099){\blacken\ellipse{72}{72}}
\put(4425,1099){\ellipse{72}{72}}
\put(4425,649){\blacken\ellipse{72}{72}}
\put(4425,649){\ellipse{72}{72}}
\put(5175,1099){\blacken\ellipse{72}{72}}
\put(5175,1099){\ellipse{72}{72}}
\put(5175,649){\blacken\ellipse{72}{72}}
\put(5175,649){\ellipse{72}{72}}
\put(5700,649){\blacken\ellipse{72}{72}}
\put(5700,649){\ellipse{72}{72}}
\put(5700,1099){\blacken\ellipse{72}{72}}
\put(5700,1099){\ellipse{72}{72}}
\put(6450,649){\blacken\ellipse{72}{72}}
\put(6450,649){\ellipse{72}{72}}
\put(6450,1099){\blacken\ellipse{72}{72}}
\put(6450,1099){\ellipse{72}{72}}
\path(3600,874)(3900,1099)
\path(3900,1099)(4425,1099)
\dashline{60.000}(4500,1099)(5100,1099)
\path(3600,874)(3900,649)
\path(3900,649)(4425,649)
\dashline{60.000}(4500,649)(5100,649)
\path(5175,649)(5700,649)
\path(5166,661)(5700,1099)
\dashline{60.000}(5775,1099)(6375,1099)
\dashline{60.000}(5775,649)(6375,649)
\put(3450,949){\makebox(0,0)[lb]{{\SetFigFont{12}{14.4}{rm}$r$}}}
\put(3750,424){\makebox(0,0)[lb]{{\SetFigFont{12}{14.4}{rm}$w_1$}}}
\put(4275,424){\makebox(0,0)[lb]{{\SetFigFont{12}{14.4}{rm}$w_2$}}}
\put(5025,424){\makebox(0,0)[lb]{{\SetFigFont{12}{14.4}{rm}$w_i$}}}
\put(5550,424){\makebox(0,0)[lb]{{\SetFigFont{12}{14.4}{rm}$w_{i+1}$}}}
\put(6300,424){\makebox(0,0)[lb]{{\SetFigFont{12}{14.4}{rm}$w_n$}}}
\put(150,874){\blacken\ellipse{72}{72}}
\put(150,874){\ellipse{72}{72}}
\put(675,874){\blacken\ellipse{72}{72}}
\put(675,874){\ellipse{72}{72}}
\put(1425,874){\blacken\ellipse{72}{72}}
\put(1425,874){\ellipse{72}{72}}
\put(1950,874){\blacken\ellipse{72}{72}}
\put(1950,874){\ellipse{72}{72}}
\put(2700,874){\blacken\ellipse{72}{72}}
\put(2700,874){\ellipse{72}{72}}
\dashline{60.000}(750,874)(1350,874)
\blacken\path(1305.000,859.000)(1350.000,874.000)(1305.000,889.000)(1305.000,859.000)
\path(1425,874)(1875,874)
\blacken\path(1830.000,859.000)(1875.000,874.000)(1830.000,889.000)(1830.000,859.000)
\path(150,874)(600,874)
\blacken\path(555.000,859.000)(600.000,874.000)(555.000,889.000)(555.000,859.000)
\dashline{60.000}(2025,874)(2625,874)
\blacken\path(2580.000,859.000)(2625.000,874.000)(2580.000,889.000)(2580.000,859.000)
\put(0,649){\makebox(0,0)[lb]{{\SetFigFont{12}{14.4}{rm}$v_1$}}}
\put(525,649){\makebox(0,0)[lb]{{\SetFigFont{12}{14.4}{rm}$v_2$}}}
\put(1275,649){\makebox(0,0)[lb]{{\SetFigFont{12}{14.4}{rm}$v_i$}}}
\put(1800,649){\makebox(0,0)[lb]{{\SetFigFont{12}{14.4}{rm}$v_{i+1}$}}}
\put(2550,649){\makebox(0,0)[lb]{{\SetFigFont{12}{14.4}{rm}$v_n$}}}
\put(5100,49){\makebox(0,0)[lb]{{\SetFigFont{12}{14.4}{rm}$T$}}}
\put(1350,49){\makebox(0,0)[lb]{{\SetFigFont{12}{14.4}{rm}$G$}}}
\end{picture}
}
}
\end{center}

Note that there is a unique non-trivial automorphism for~$T$:
exchange $u_{i+1}$ and $w_{i+1}$, $\dots$, $u_n$ and $w_n$,
and map the other vertices onto themselves.
We construct two trees~$T_1$ and~$T_2$ from~$T$.
With respect to~$T$,
tree~$T_1$ has two extra nodes~$x_0,x_1$ which are connected with node~$u_j$,
and~$T_2$ has  extra nodes~$y_0,y_1$ which are connected with node~$w_j$.
The extra edges enforce that an isomorphism between~$T_1$ and~$T_2$
has to map~$u_j$ to~$w_j$, because these are the only nodes of degree~4 (for $j < n$).
Now,
if $v_i < v_j$, then the above automorphism of~$T$ yields an isomorphism between~$T_1$ and~$T_2$.
On the other hand,
if $v_i \geq v_j$, then there is no isomorphism between~$T_1$ and~$T_2$.

We modify~$T$ to a graph~$H$ that is no longer a tree,
but planar and 3-connected.
Split each node~$v$ of degree~1 or~2 in~$T$ into three nodes $v^0, v^1, v^2$.
Connect these nodes via edges $(v^0,v^1)$ and $(v^1,v^2)$.
If~$v$ has degree~1, then additionally put  the edge $(v^0,v^2)$. 
Now, if $(u,v)$ is an edge in~$T$, where~$u$ and~$v$ have degree~1 or~2,
then we have edges $(u^0,v^0)$, $(u^1,v^1)$, and $(u^2,v^2)$ in~$H$.
The following picture illustrates the situation.
In~(a), node~$v$ has degree~2,
in~(b), node~$v$ has degree~1.

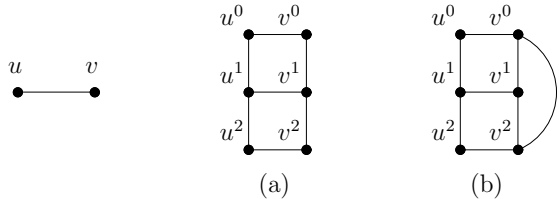
\begin{center}
\scalebox{0.8}{ 
\setlength{\unitlength}{0.00083333in}
\begingroup\makeatletter\ifx\SetFigFont\undefined
\def\x#1#2#3#4#5#6#7\relax{\def\x{#1#2#3#4#5#6}}%
\expandafter\x\fmtname xxxxxx\relax \def\y{splain}%
\ifx\x\y   
\gdef\SetFigFont#1#2#3{%
  \ifnum #1<17\tiny\else \ifnum #1<20\small\else
  \ifnum #1<24\normalsize\else \ifnum #1<29\large\else
  \ifnum #1<34\Large\else \ifnum #1<41\LARGE\else
     \huge\fi\fi\fi\fi\fi\fi
  \csname #3\endcsname}%
\else
\gdef\SetFigFont#1#2#3{\begingroup
  \count@#1\relax \ifnum 25<\count@\count@25\fi
  \def\x{\endgroup\@setsize\SetFigFont{#2pt}}%
  \expandafter\x
    \csname \romannumeral\the\count@ pt\expandafter\endcsname
    \csname @\romannumeral\the\count@ pt\endcsname
  \csname #3\endcsname}%
\fi
\fi\endgroup
{\renewcommand{\dashlinestretch}{30}
\begin{picture}(4284,1551)(0,-10)
\put(1875,1329){\blacken\ellipse{72}{72}}
\put(1875,1329){\ellipse{72}{72}}
\put(1875,429){\blacken\ellipse{72}{72}}
\put(1875,429){\ellipse{72}{72}}
\put(1875,879){\blacken\ellipse{72}{72}}
\put(1875,879){\ellipse{72}{72}}
\put(2325,879){\blacken\ellipse{72}{72}}
\put(2325,879){\ellipse{72}{72}}
\put(2325,1329){\blacken\ellipse{72}{72}}
\put(2325,1329){\ellipse{72}{72}}
\put(2325,429){\blacken\ellipse{72}{72}}
\put(2325,429){\ellipse{72}{72}}
\path(2325,879)(2325,429)
\path(2325,1329)(2325,879)
\path(1875,429)(2325,429)
\path(1875,879)(1875,429)
\path(1875,1329)(2325,1329)
\path(1875,1329)(1875,879)
\path(1875,879)(2325,879)
\put(1650,1404){\makebox(0,0)[lb]{{\SetFigFont{12}{14.4}{rm}$u^0$}}}
\put(1650,954){\makebox(0,0)[lb]{{\SetFigFont{12}{14.4}{rm}$u^1$}}}
\put(1650,504){\makebox(0,0)[lb]{{\SetFigFont{12}{14.4}{rm}$u^2$}}}
\put(2100,1404){\makebox(0,0)[lb]{{\SetFigFont{12}{14.4}{rm}$v^0$}}}
\put(2100,954){\makebox(0,0)[lb]{{\SetFigFont{12}{14.4}{rm}$v^1$}}}
\put(2100,504){\makebox(0,0)[lb]{{\SetFigFont{12}{14.4}{rm}$v^2$}}}
\put(3525,1329){\blacken\ellipse{72}{72}}
\put(3525,1329){\ellipse{72}{72}}
\put(3525,429){\blacken\ellipse{72}{72}}
\put(3525,429){\ellipse{72}{72}}
\put(3525,879){\blacken\ellipse{72}{72}}
\put(3525,879){\ellipse{72}{72}}
\put(3975,879){\blacken\ellipse{72}{72}}
\put(3975,879){\ellipse{72}{72}}
\put(3975,1329){\blacken\ellipse{72}{72}}
\put(3975,1329){\ellipse{72}{72}}
\put(3975,429){\blacken\ellipse{72}{72}}
\put(3975,429){\ellipse{72}{72}}
\path(3975,879)(3975,429)
\path(3975,1329)(3975,879)
\path(3525,429)(3975,429)
\path(3525,879)(3525,429)
\path(3525,1329)(3975,1329)
\path(3525,1329)(3525,879)
\path(3525,879)(3975,879)
\put(3300,1404){\makebox(0,0)[lb]{{\SetFigFont{12}{14.4}{rm}$u^0$}}}
\put(3300,954){\makebox(0,0)[lb]{{\SetFigFont{12}{14.4}{rm}$u^1$}}}
\put(3300,504){\makebox(0,0)[lb]{{\SetFigFont{12}{14.4}{rm}$u^2$}}}
\put(3750,1404){\makebox(0,0)[lb]{{\SetFigFont{12}{14.4}{rm}$v^0$}}}
\put(3750,954){\makebox(0,0)[lb]{{\SetFigFont{12}{14.4}{rm}$v^1$}}}
\put(3750,504){\makebox(0,0)[lb]{{\SetFigFont{12}{14.4}{rm}$v^2$}}}
\put(3787.500,879.000){\arc{975.000}{5.1072}{7.4592}}
\put(75,879){\blacken\ellipse{72}{72}}
\put(75,879){\ellipse{72}{72}}
\put(675,879){\blacken\ellipse{72}{72}}
\put(675,879){\ellipse{72}{72}}
\path(75,879)(675,879)
\put(0,1029){\makebox(0,0)[lb]{{\SetFigFont{12}{14.4}{rm}$u$}}}
\put(600,1029){\makebox(0,0)[lb]{{\SetFigFont{12}{14.4}{rm}$v$}}}
\put(1950,54){\makebox(0,0)[lb]{{\SetFigFont{12}{14.4}{rm}(a)}}}
\put(3600,54){\makebox(0,0)[lb]{{\SetFigFont{12}{14.4}{rm}(b)}}}
\end{picture}
}
}
\end{center}

A special case is node~$w_i$ which has degree~3.
For~$w_i$ we need a gadget with~9 nodes which are connected as a~$3 \times 3$ grid.
The connections from this graph gadget (bold lines) to the other nodes are shown in the following picture.

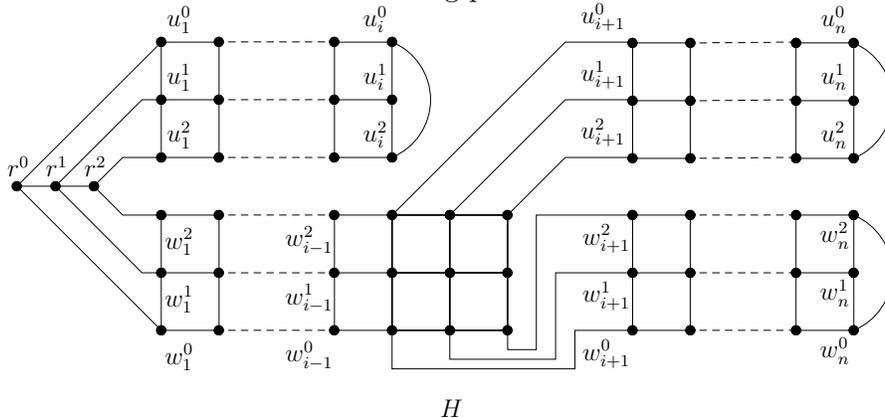
\begin{center}
\scalebox{0.8}{ 
\setlength{\unitlength}{0.00083333in}
\begingroup\makeatletter\ifx\SetFigFont\undefined
\def\x#1#2#3#4#5#6#7\relax{\def\x{#1#2#3#4#5#6}}%
\expandafter\x\fmtname xxxxxx\relax \def\y{splain}%
\ifx\x\y   
\gdef\SetFigFont#1#2#3{%
  \ifnum #1<17\tiny\else \ifnum #1<20\small\else
  \ifnum #1<24\normalsize\else \ifnum #1<29\large\else
  \ifnum #1<34\Large\else \ifnum #1<41\LARGE\else
     \huge\fi\fi\fi\fi\fi\fi
  \csname #3\endcsname}%
\else
\gdef\SetFigFont#1#2#3{\begingroup
  \count@#1\relax \ifnum 25<\count@\count@25\fi
  \def\x{\endgroup\@setsize\SetFigFont{#2pt}}%
  \expandafter\x
    \csname \romannumeral\the\count@ pt\expandafter\endcsname
    \csname @\romannumeral\the\count@ pt\endcsname
  \csname #3\endcsname}%
\fi
\fi\endgroup
{\renewcommand{\dashlinestretch}{30}
\begin{picture}(7029,3230)(0,-10)
\put(6344,3047){\makebox(0,0)[lb]{{\SetFigFont{12}{14.4}{rm}$u_n^0$}}}
\put(6344,2597){\makebox(0,0)[lb]{{\SetFigFont{12}{14.4}{rm}$u_n^1$}}}
\put(6344,2147){\makebox(0,0)[lb]{{\SetFigFont{12}{14.4}{rm}$u_n^2$}}}
\put(4474,2159){\makebox(0,0)[lb]{{\SetFigFont{12}{14.4}{rm}$u_{i+1}^2$}}}
\put(4474,2609){\makebox(0,0)[lb]{{\SetFigFont{12}{14.4}{rm}$u_{i+1}^1$}}}
\put(4474,3059){\makebox(0,0)[lb]{{\SetFigFont{12}{14.4}{rm}$u_{i+1}^0$}}}
\put(6337,912){\makebox(0,0)[lb]{{\SetFigFont{12}{14.4}{rm}$w_n^1$}}}
\put(6337,1362){\makebox(0,0)[lb]{{\SetFigFont{12}{14.4}{rm}$w_n^2$}}}
\put(6337,462){\makebox(0,0)[lb]{{\SetFigFont{12}{14.4}{rm}$w_n^0$}}}
\put(4485,882){\makebox(0,0)[lb]{{\SetFigFont{12}{14.4}{rm}$w_{i+1}^1$}}}
\put(4485,432){\makebox(0,0)[lb]{{\SetFigFont{12}{14.4}{rm}$w_{i+1}^0$}}}
\put(4485,1332){\makebox(0,0)[lb]{{\SetFigFont{12}{14.4}{rm}$w_{i+1}^2$}}}
\put(2167,867){\makebox(0,0)[lb]{{\SetFigFont{12}{14.4}{rm}$w_{i-1}^1$}}}
\put(2167,1317){\makebox(0,0)[lb]{{\SetFigFont{12}{14.4}{rm}$w_{i-1}^2$}}}
\put(2167,417){\makebox(0,0)[lb]{{\SetFigFont{12}{14.4}{rm}$w_{i-1}^0$}}}
\put(1228,1313){\makebox(0,0)[lb]{{\SetFigFont{12}{14.4}{rm}$w_1^2$}}}
\put(1228,863){\makebox(0,0)[lb]{{\SetFigFont{12}{14.4}{rm}$w_1^1$}}}
\put(1228,413){\makebox(0,0)[lb]{{\SetFigFont{12}{14.4}{rm}$w_1^0$}}}
\put(6412.500,1174.000){\arc{975.000}{5.1072}{7.4592}}
\put(6412.500,2517.000){\arc{975.000}{5.1072}{7.4592}}
\put(2812.500,2524.000){\arc{975.000}{5.1072}{7.4592}}
\put(6150,1624){\blacken\ellipse{72}{72}}
\put(6150,1624){\ellipse{72}{72}}
\put(6150,1174){\blacken\ellipse{72}{72}}
\put(6150,1174){\ellipse{72}{72}}
\put(6150,724){\blacken\ellipse{72}{72}}
\put(6150,724){\ellipse{72}{72}}
\put(6600,724){\blacken\ellipse{72}{72}}
\put(6600,724){\ellipse{72}{72}}
\put(6600,1174){\blacken\ellipse{72}{72}}
\put(6600,1174){\ellipse{72}{72}}
\put(6600,1624){\blacken\ellipse{72}{72}}
\put(6600,1624){\ellipse{72}{72}}
\put(6150,2067){\blacken\ellipse{72}{72}}
\put(6150,2067){\ellipse{72}{72}}
\put(6150,2517){\blacken\ellipse{72}{72}}
\put(6150,2517){\ellipse{72}{72}}
\put(6150,2967){\blacken\ellipse{72}{72}}
\put(6150,2967){\ellipse{72}{72}}
\put(6600,2067){\blacken\ellipse{72}{72}}
\put(6600,2067){\ellipse{72}{72}}
\put(6600,2517){\blacken\ellipse{72}{72}}
\put(6600,2517){\ellipse{72}{72}}
\put(6600,2967){\blacken\ellipse{72}{72}}
\put(6600,2967){\ellipse{72}{72}}
\put(3000,1174){\blacken\ellipse{72}{72}}
\put(3000,1174){\ellipse{72}{72}}
\put(3000,1624){\blacken\ellipse{72}{72}}
\put(3000,1624){\ellipse{72}{72}}
\put(3000,724){\blacken\ellipse{72}{72}}
\put(3000,724){\ellipse{72}{72}}
\put(3450,724){\blacken\ellipse{72}{72}}
\put(3450,724){\ellipse{72}{72}}
\put(3450,1174){\blacken\ellipse{72}{72}}
\put(3450,1174){\ellipse{72}{72}}
\put(3450,1624){\blacken\ellipse{72}{72}}
\put(3450,1624){\ellipse{72}{72}}
\put(3900,1624){\blacken\ellipse{72}{72}}
\put(3900,1624){\ellipse{72}{72}}
\put(3900,1174){\blacken\ellipse{72}{72}}
\put(3900,1174){\ellipse{72}{72}}
\put(3900,724){\blacken\ellipse{72}{72}}
\put(3900,724){\ellipse{72}{72}}
\put(4875,1624){\blacken\ellipse{72}{72}}
\put(4875,1624){\ellipse{72}{72}}
\put(4875,1174){\blacken\ellipse{72}{72}}
\put(4875,1174){\ellipse{72}{72}}
\put(4875,724){\blacken\ellipse{72}{72}}
\put(4875,724){\ellipse{72}{72}}
\put(5325,1174){\blacken\ellipse{72}{72}}
\put(5325,1174){\ellipse{72}{72}}
\put(5325,1624){\blacken\ellipse{72}{72}}
\put(5325,1624){\ellipse{72}{72}}
\put(5325,724){\blacken\ellipse{72}{72}}
\put(5325,724){\ellipse{72}{72}}
\put(2550,1174){\blacken\ellipse{72}{72}}
\put(2550,1174){\ellipse{72}{72}}
\put(2550,1624){\blacken\ellipse{72}{72}}
\put(2550,1624){\ellipse{72}{72}}
\put(3000,2974){\blacken\ellipse{72}{72}}
\put(3000,2974){\ellipse{72}{72}}
\put(3000,2524){\blacken\ellipse{72}{72}}
\put(3000,2524){\ellipse{72}{72}}
\put(3000,2074){\blacken\ellipse{72}{72}}
\put(3000,2074){\ellipse{72}{72}}
\put(2550,2074){\blacken\ellipse{72}{72}}
\put(2550,2074){\ellipse{72}{72}}
\put(2550,2524){\blacken\ellipse{72}{72}}
\put(2550,2524){\ellipse{72}{72}}
\put(2550,2974){\blacken\ellipse{72}{72}}
\put(2550,2974){\ellipse{72}{72}}
\put(4875,2067){\blacken\ellipse{72}{72}}
\put(4875,2067){\ellipse{72}{72}}
\put(4875,2517){\blacken\ellipse{72}{72}}
\put(4875,2517){\ellipse{72}{72}}
\put(4875,2967){\blacken\ellipse{72}{72}}
\put(4875,2967){\ellipse{72}{72}}
\put(5325,2517){\blacken\ellipse{72}{72}}
\put(5325,2517){\ellipse{72}{72}}
\put(5325,2967){\blacken\ellipse{72}{72}}
\put(5325,2967){\ellipse{72}{72}}
\put(5325,2067){\blacken\ellipse{72}{72}}
\put(5325,2067){\ellipse{72}{72}}
\put(2550,724){\blacken\ellipse{72}{72}}
\put(2550,724){\ellipse{72}{72}}
\put(1650,724){\blacken\ellipse{72}{72}}
\put(1650,724){\ellipse{72}{72}}
\put(1650,1174){\blacken\ellipse{72}{72}}
\put(1650,1174){\ellipse{72}{72}}
\put(1650,1624){\blacken\ellipse{72}{72}}
\put(1650,1624){\ellipse{72}{72}}
\put(1650,2074){\blacken\ellipse{72}{72}}
\put(1650,2074){\ellipse{72}{72}}
\put(1650,2524){\blacken\ellipse{72}{72}}
\put(1650,2524){\ellipse{72}{72}}
\put(1650,2974){\blacken\ellipse{72}{72}}
\put(1650,2974){\ellipse{72}{72}}
\put(1200,1174){\blacken\ellipse{72}{72}}
\put(1200,1174){\ellipse{72}{72}}
\put(1200,1624){\blacken\ellipse{72}{72}}
\put(1200,1624){\ellipse{72}{72}}
\put(1200,724){\blacken\ellipse{72}{72}}
\put(1200,724){\ellipse{72}{72}}
\put(75,1849){\blacken\ellipse{72}{72}}
\put(75,1849){\ellipse{72}{72}}
\put(375,1849){\blacken\ellipse{72}{72}}
\put(375,1849){\ellipse{72}{72}}
\put(675,1849){\blacken\ellipse{72}{72}}
\put(675,1849){\ellipse{72}{72}}
\put(1200,2074){\blacken\ellipse{72}{72}}
\put(1200,2074){\ellipse{72}{72}}
\put(1200,2524){\blacken\ellipse{72}{72}}
\put(1200,2524){\ellipse{72}{72}}
\put(1200,2974){\blacken\ellipse{72}{72}}
\put(1200,2974){\ellipse{72}{72}}
\path(6150,1624)(6150,724)
\path(6150,1624)(6600,1624)
\path(6600,1624)(6600,724)
\path(6600,724)(6150,724)
\path(6150,1174)(6600,1174)
\path(6150,2967)(6150,2067)
\path(6150,2067)(6600,2067)
\path(6600,2067)(6600,2967)
\path(6600,2967)(6150,2967)
\path(6150,2517)(6600,2517)
\thicklines
\path(3000,1624)(3450,1624)
\path(3450,1624)(3900,1624)
\path(3000,1174)(3900,1174)
\path(3000,724)(3900,724)
\path(3450,1624)(3450,724)
\thinlines
\path(3900,724)(3900,574)(4125,574)
	(4125,1624)(4875,1624)
\path(3450,724)(3450,499)(4275,499)
	(4275,1174)(4875,1174)
\path(3450,1624)(4350,2524)(4875,2524)
\path(3000,1624)(4350,2974)(4875,2974)
\dashline{60.000}(5400,1624)(6150,1624)
\dashline{60.000}(5400,1174)(6150,1174)
\dashline{60.000}(5400,724)(6150,724)
\path(5325,1174)(5325,724)
\path(5325,1624)(5325,1174)
\path(4875,724)(5325,724)
\path(4875,1174)(4875,724)
\path(4875,1624)(4875,1174)
\path(4875,1174)(5325,1174)
\path(2550,1624)(3000,1624)
\thicklines
\path(3000,1624)(3000,724)
\thinlines
\path(3000,724)(2550,724)
\path(2550,1174)(3000,1174)
\dashline{60.000}(1725,1174)(2475,1174)
\dashline{60.000}(1725,1624)(2475,1624)
\path(2550,1624)(2550,724)
\path(4875,1624)(5325,1624)
\dashline{60.000}(1725,2074)(2475,2074)
\dashline{60.000}(1725,2524)(2475,2524)
\dashline{60.000}(1725,2974)(2475,2974)
\path(3900,1624)(4350,2074)(4875,2074)
\path(3000,2974)(3000,2074)
\path(3000,2974)(2550,2974)
\path(2550,2974)(2550,2074)
\path(2550,2074)(3000,2074)
\path(2550,2524)(3000,2524)
\thicklines
\path(3900,1624)(3900,724)
\thinlines
\dashline{60.000}(5400,2967)(6150,2974)
\dashline{60.000}(5400,2517)(6150,2524)
\dashline{60.000}(5400,2067)(6150,2074)
\path(5325,2517)(5325,2067)
\path(5325,2967)(5325,2517)
\path(4875,2067)(5325,2067)
\path(4875,2517)(4875,2067)
\path(4875,2967)(5325,2967)
\path(4875,2967)(4875,2517)
\path(4875,2517)(5325,2517)
\dashline{60.000}(1725,724)(2475,724)
\path(1200,2974)(1650,2974)
\path(1650,2974)(1650,2074)
\path(1650,2074)(1200,2074)
\path(1200,2524)(1650,2524)
\path(1200,1624)(1650,1624)
\path(1650,1624)(1650,724)
\path(1650,724)(1200,724)
\path(1650,1174)(1200,1174)
\path(75,1849)(675,1849)
\path(1200,1174)(1200,724)
\path(1200,1624)(1200,1174)
\path(1200,2074)(1200,2974)
\path(1200,1624)(900,1624)(675,1849)
	(900,2074)(1200,2074)
\path(1200,1174)(1050,1174)(375,1849)
	(1050,2524)(1200,2524)
\path(1200,724)(75,1849)(1200,2974)
\path(3000,724)(3000,424)(4425,424)
	(4425,724)(4875,724)
\put(300,1924){\makebox(0,0)[lb]{{\SetFigFont{12}{14.4}{rm}$r^1$}}}
\put(600,1924){\makebox(0,0)[lb]{{\SetFigFont{12}{14.4}{rm}$r^2$}}}
\put(0,1924){\makebox(0,0)[lb]{{\SetFigFont{12}{14.4}{rm}$r^0$}}}
\put(1245,2146){\makebox(0,0)[lb]{{\SetFigFont{12}{14.4}{rm}$u_1^2$}}}
\put(1245,2596){\makebox(0,0)[lb]{{\SetFigFont{12}{14.4}{rm}$u_1^1$}}}
\put(1245,3046){\makebox(0,0)[lb]{{\SetFigFont{12}{14.4}{rm}$u_1^0$}}}
\put(2775,2149){\makebox(0,0)[lb]{{\SetFigFont{12}{14.4}{rm}$u_i^2$}}}
\put(2775,2599){\makebox(0,0)[lb]{{\SetFigFont{12}{14.4}{rm}$u_i^1$}}}
\put(2775,3049){\makebox(0,0)[lb]{{\SetFigFont{12}{14.4}{rm}$u_i^0$}}}
\put(3375,49){\makebox(0,0)[lb]{{\SetFigFont{12}{14.4}{rm}$H$}}}
\end{picture}
}
}
\end{center}

Now it suffices again to mark the nodes corresponding to~$v_j$.
That is,
define graph~$H_1$ as graph~$H$ plus the  edge $(u_j^0,u_j^2)$,
and~$H_2$ as~$H$ plus the  edge $(w_j^0,w_j^2)$.
Note that~$H_1$ and~$H_2$ are planar and 3-connected.
Furthermore,
any isomorphism between~$H_1$ and~$H_2$ has to map  
$u_j^0$ to $w_j^0$,
$u_j^1$ to $w_j^1$, and
$u_j^2$ to $w_j^2$.
Again,
this is only possible iff $v_i < v_j$.
This completes the proof of Theorem~\ref{T:Planar3ConnGIHardforL}.

\vspace{2ex}

A final observation is about oriented trees.
An {\em oriented tree\/} is a tree with a planar rotation scheme.
It is not hard to see that one can adapt Lindell's algorithm to work for oriented trees,
so that the corresponding isomorphism problem is in~$\L$.
We show that it is also hard for~$\L$.

\begin{theorem}
\label{T:OrientedTreeIsoInL}
Oriented tree isomorphism is complete for $\L$.
\end{theorem}

We reduce $\Ord$ to the oriented tree isomorphism problem.
Let~$G$ be the given line graph and
consider again the trees~$T_1$ and~$T_2$ from above constructed from~$G$
in the proof of Theorem~\ref{T:Planar3ConnGIHardforL}.
For nodes of degree~1 or~2 there is only one rotation scheme.
Therefore we only have to take care of the nodes of degree~3 and~4,
i.e.\ $w_i$, $w_j$, and~$u_j$.
\begin{itemize}
\item
The rotation scheme for~$w_i$ is easy to handle:
output the edges around~$w_i$ for~$T_1$ in an arbitrary order,
and choose the opposite order for~$w_i$ in~$T_2$.
This definition fits together with
the only possible isomorphism that should exchange~$u_{i+1}$ and~$w_{i+1}$.
\item
In the rotation scheme for~$w_j$ the order of edges to the neighbors can be chosen as
$w_{j-1},~y_0,~w_{j+1},~y_1$,
and around $u_j$ in order $u_{j-1},~x_0,~u_{j+1},~x_1$.
Because of the symmetry of the parts $(u_j,x_0)$ and  $(u_j,x_1)$ in~$T_1$
and of $(w_j,y_0)$ and  $(w_j,y_1)$ in~$T_2$
an isomorphism mapping~$w_j$ to~$u_j$ can be defined respecting the 
rotation schemes for these nodes.
\end{itemize}
Now the same argument as for Theorem~\ref{T:Planar3ConnGIHardforL}
shows that the oriented trees~$T_1$ and~$T_2$ are isomorphic iff $v_i < v_j$.
This proves the theorem.


\vskip-0.3cm
\section*{Open Problems}

The most challenging task is to close the gap between $\L$ and $\UL \cap \co\UL$
for the planar 3-connected graph isomorphism problem.
Another goal is to extend the isomorphism test to arbitrary planar graphs.
If the graph is not connected,
we can compute the connected components and consider them separately.
Hence, we may assume that the graph is connected.
Then we can determine the articulation points and the separating pairs
and get the 1- and 2-connected components of the graph.
For sequential algorithms to compute a canonical form for these graphs
see for example~\cite{KukHolCoo04}.
Miller and Reif~\cite{MilRei91} provide an $\AC^1$-reduction from 
planar graphs to planar 3-connected graphs.
We ask whether one can compute a canonical form for planar graphs in (unambiguous) logspace.

\vskip-0.3cm
\subsection*{Acknowledgment} 
We thank Jacobo Tor{\'a}n and the anonymous referees for helpful comments on the manuscript.

\bibliographystyle{alpha}

\end{document}